\documentclass[a4paper,USenglish,cleveref,autoref,thm-restate]{lipics-v2021}

\usepackage[utf8]{inputenc}
\usepackage[T1]{fontenc}
\usepackage{amsmath,amsfonts,amsthm,amssymb}

\usepackage{footnote}
\makesavenoteenv{tabular}
\makesavenoteenv{table}

\newcommand*{\Otilde}{\widetilde{O}}

\DeclareMathOperator{\diam}{diam}
\DeclareMathOperator{\ecc}{ecc}

\newcommand*{\emDash}{\hspace*{.25pt}--\hspace*{.25pt}}

\newcommand*{\nwspace}{\hspace*{.1em}} 

\let\oldsqrt\sqrt
\def\hksqrt{\mathpalette\DHLhksqrt}
\def\DHLhksqrt#1#2{\setbox0=\hbox{$#1\oldsqrt{#2\,}$}\dimen0=\ht0
   \advance\dimen0-0.2\ht0
   \setbox2=\hbox{\vrule height\ht0 depth -\dimen0}%
   {\box0\lower0.4pt\box2}}
\renewcommand\sqrt\hksqrt

\newcounter{property}										
\crefname{property}{Property}{Properties}					
\makeatletter
\newcommand*{\inlineequation}[2][]{%
  \begingroup
    \refstepcounter{property}%
    \ifx\\#1\\%
    \else
      \label{#1}%
    \fi
    \relpenalty=10000 %
    \binoppenalty=10000 %
    \ensuremath{%
      #2%
    }%
    ~\@eqnnum
  \endgroup
}
\makeatother

\renewcommand{\leq}{\leqslant}
\renewcommand{\geq}{\geqslant}
\renewcommand{\le}{\leqslant}
\renewcommand{\ge}{\geqslant}

\author{Davide Bilò}%
{Department of Humanities and Social Sciences, University of Sassari, Italy}%
{davidebilo@uniss.it}%
{0000-0003-3169-4300} 
{This work was partially supported by the Research Grant FBS2016\_BILO, funded by ``Fondazione di Sardegna'' in 2016.} 

\author{Sarel Cohen}%
{Hasso Plattner Institute, University of Potsdam, Germany}%
{sarel.cohen@hpi.de}%
{} 
{} 

\author{Tobias Friedrich}%
{Hasso Plattner Institute, University of Potsdam, Germany}%
{tobias.friedrich@hpi.de}%
{0000-0003-0076-6308} 
{} 

\author{Martin Schirneck}%
{Hasso Plattner Institute, University of Potsdam, Germany}%
{martin.schirneck@hpi.de}%
{} 
{} 

\authorrunning{D.~Bilò, S.~Cohen, T.~Friedrich, and M.~Schirneck}
\Copyright{Davide Bilò, Sarel Cohen, Tobias Friedrich, and Martin Schirneck}

\title{Space-Efficient Fault-Tolerant Diameter Oracles}



\nolinenumbers 

\hideLIPIcs  

\begin{document}

\maketitle

\begin{abstract}
We design \emph{$f$-edge fault-tolerant diameter oracles} ($f$-FDO, or simply FDO if $f=1$). For a given directed or undirected and possibly edge-weighted graph $G$ with $n$ vertices and $m$ edges and a positive integer $f$, we preprocess the graph and construct a data structure that, when queried with a set $F$ of edges, where $|F| \leq f$, returns the diameter of $G\,{-}\,F$. An $f$-FDO has stretch $\sigma \ge 1$ if the returned value $\widehat D$ satisfies $\diam(G\,{-}\,F) \le \widehat D \le \sigma \diam(G\,{-}\,F)$.

For the case of a single edge failure ($f=1$) in an unweighted directed graph, 
there exists an approximate FDO by Henzinger et al.\ [ITCS 2017] with stretch $(1+\varepsilon)$, constant query time, space $O(m)$, and a combinatorial preprocessing time of $\widetilde{O}(mn + n^{1.5} \sqrt{Dm/\varepsilon})$, where $D$ is the diameter.

We present an FDO for directed graphs with the same stretch, query time, and space. 
It has a preprocessing time of $\widetilde{O}(mn + n^2/\varepsilon)$, which is better for constant $\varepsilon > 0$. The preprocessing time nearly matches a conditional lower bound for combinatorial algorithms, also by Henzinger et al. With fast matrix multiplication, we achieve a preprocessing time of $\widetilde{O}(n^{2.5794} + n^2/\varepsilon)$.
We further prove an information-theoretic lower bound showing that any FDO with stretch better than $3/2$ requires $\Omega(m)$ bits of space. Thus, for constant $0 < \varepsilon < 3/2$, our combinatorial $(1+ \varepsilon)$-approximate FDO is near-optimal in all parameters.

In the case of multiple edge failures ($f>1$) in undirected graphs with non-negative edge weights, we give an $f$-FDO with stretch $(f+2)$, query time $O(f^2\log^2{n})$, $\widetilde{O}(fn)$ space, and preprocessing time $\widetilde{O}(fm)$. We complement this with a lower bound excluding any finite stretch in $o(fn)$ space.

Many real-world networks have polylogarithmic diameter. We show that for those graphs and up to $f = o(\log n/ \log\log n)$ failures one can swap approximation for query time and space. We present an exact combinatorial $f$-FDO with preprocessing time $mn^{1+o(1)}$, query time $n^{o(1)}$, and space $n^{2+o(1)}$. When using fast matrix multiplication instead, the preprocessing time can be improved to $n^{\omega+o(1)}$,
where $\omega < 2.373$ is the matrix multiplication exponent.
\end{abstract}

\section{Introduction}
\label{sec:intro}

The diameter is one of the most fundamental graph parameters.
It plays a particular significant role in the analysis of communication networks
as the time to transmit a message to all nodes is strongly related with the diameter.
Several lines of work have recently attacked the problem of computing the diameter 
in different settings. 
For example, Choudhary and Gold~\cite{ChoudharyG20} constructed diameter spanners,
which are subgraphs that approximately preserve the diameter of the original graph,
Ancona et al.~\cite{AnconaHRWW19} developed algorithms for computing the diameter in dynamic scenarios and proved matching conditional lower bounds,
and Bonnet~\cite{Bonnet21} proved that, 
for any constant $\varepsilon > 0$, computing a $(7/4 -\varepsilon)$-approximation of the diameter of a sparse graph $n$ vertices and $m=n^{1+o(1)}$ edges requires $m^{4/3-o(1)}$ time,
unless the Strong Exponential Time Hypothesis fails.

In this paper, we approach the diameter from the perspective of fault tolerance.
A communication network may be subject to a small number of transient failures,
and we want to quickly find out the new diameter without recomputing it from scratch.
Therefore, we study the problem of constructing space-efficient data structures that can quickly report the diameter even if up to $f$ edges fail in the graph.
We refer to them as \emph{$f$-edge fault-tolerant diameter oracles} 
($f$-FDO, or simply FDO if $f=1$).
More precisely, given an undirected or directed and possibly edge-weighted graph $G$
and a positive integer $f$, we want to construct an $f$-FDO that, when queried on a set $F$ of up to $f$ edges of $G$, returns a value $\widehat D$ that is always at least as large the diameter of $G\,{-}\,F$,
denoted by $\diam(G\,{-}\,F)$.
We say that an $f$-FDO has a \emph{stretch} of $\sigma \ge 1$ (or that it is \emph{$\sigma$-approximate})
if the value $\widehat D$ returned by the oracle 
additionally satisfies $\diam(G\,{-}\,F) \leq  \widehat D \leq \sigma \diam(G\,{-}\,F)$. 

When designing $f$-FDOs one must find a good compromise between the following parameters: the stretch, the time needed to query the oracle, the size of the data structure,
and the preprocessing time needed to build it.
We focus particularly on space-optimal solutions, while keeping the query and preprocessing times low.
For the case of a single edge failure in undirected edge-weighted graphs, 
there are two folklore FDOs known.
One reports the exact diameter and has size $O(m)$,
while the other takes $O(n)$ space, but guarantees only a stretch of $2$.
(more details are given in \Cref{subsec:prelims_trivial}.)
In a sense they mark the extreme points of a spectrum.
It is natural to ask whether there are more trade-offs possible between the stretch and size of an FDO.
More precisely, we pose the following question.

\textbf{\textsf{Question 1 - space vs.\ approximation trade-off.}}
What is the minimum achievable size of an FDO for a given stretch $\sigma$?
To answer the question, we prove an information-theoretic lower bound.
It shows that for undirected unweighted graphs and every (even non-constant) $1 \le \sigma < 3/2$,
every $\sigma$-approximate diameter oracle requires $\Omega(m)$ bits of space.
The space lower bound also holds for the harder case of directed graphs.
The size of the exact folklore FDO is thus optimal up to the size of a machine word.
Moreover, we prove that the stretch $2$ of the approximate FDO
cannot be improved on weighted graphs while keeping $O(n)$ space.  

\begin{restatable}{theorem}{spacelowerboundsingle}
\label{thm:space_lower_bound_single}
	Any FDO with stretch $\sigma = \sigma(m) < 3/2$	must take $\Omega(m)$ bits of space on 
	undirected graphs with $m$ edges.
	The bound increases to $\sigma < 2$ if the graphs are edge-weighted.
\end{restatable}

When we focus our attention on the preprocessing time,
the exact FDO can be constructed in $\Otilde(n^3)$ time\footnote{
	For a positive function $g(n,m,f)$, we use $\Otilde(g)$ 
	to denote $O(g \cdot \textsf{polylog}(n))$.
}
using the distance sensitivity oracle (DSO) of Bernstein and Karger~\cite{BeKa09}.
Henzinger et al.~\cite{HenzingerL0W17} proved an essentially matching conditional lower bound
for combinatorial\footnote{%
    The term ``combinatorial algorithm'' is not well-defined, and is often interpreted as not using any matrix multiplication.
    Arguably, combinatorial algorithms can be considered 
    efficient in practice as the constants hidden in the matrix
    multiplication bounds are rather high.
}
algorithms.
They assumed that any combinatorial algorithm requires $n^{3-o(1)}$ time to multiply
two Boolean $n \times n$ matrices, known as the BMM conjecture.
The restriction to combinatorial algorithms is crucial
as the task is reducible to integer matrices and one can use
fast matrix multiplication to solve it in $O(n^\omega)$ time,
where $\omega<2.37286$ is the matrix multiplication exponent~\cite{AlmanVWilliams21RefinedLaserMethod}.
Under the BMM conjecture, Henzinger~et~al.~\cite{HenzingerL0W17} showed that,
for $0 < \varepsilon < 1/3$, any combinatorial preprocessing algorithm 
requires $n^{3-o(1)}$ time to build an FDO of stretch $(1+\varepsilon)$,
even if we allow $O(n^{2-\delta})$ query time for any constant $\delta > 0$.

They match this bound with an FDO with stretch $(1+\varepsilon)$ and $O(1)$ query time
that can be constructed in time $\Otilde(mn + n^{1.5}\sqrt{\diam(G) \cdot m/\varepsilon})$.
Their oracle also reports the radius and vertex eccentricities
in the presence of a single edge failure.
Even on sparse graphs with $m = \Otilde(n)$ edges and constant diameter, 
the preprocessing time is $\Otilde(n^{2.5}/\sqrt{\varepsilon})$.
For constant $\varepsilon > 0$, this is by a factor $\sqrt{n}$ larger than
the $\Otilde(mn)$ time needed to build the DSO of Bernstein and Karger~\cite{BeKa09}.
It is interesting whether one can close the gap.

\textbf{\textsf{Question 2 - fast preprocessing time.}} 
Does there exist a combinatorial algorithm that constructs in $\Otilde(mn)$ time
an FDO with stretch $(1+\varepsilon)$ and constant query time? 
In addition, can one bypass the combinatorial lower bound
by using fast matrix multiplication?
We answer these questions affirmatively for the diameter case with the following theorem.
The proof of the algebraic part uses the DSO
presented very recently by Gu and Ren~\cite{GuRen21ConstructingDSO_ICALP}. 

\begin{restatable}{theorem}{singlefailure}
\label{thm:single_failure}
	For every unweighted directed graph and $\varepsilon > 0$,
	there exists a randomized combinatorial $(1+\varepsilon)$-approximate FDO
	that takes $O(m)$ space and has $\Otilde(mn + n^2/\varepsilon)$ preprocessing time
	and $O(1)$ query time.
	The returned values are correct w.h.p.\footnote{%
		An event occurs \emph{with high probability} (w.h.p.)
		if it has probability at least $1 - n^{-c}$ for some $c > 0$.
	}
    Using fast matrix multiplication instead, one can construct the FDO in time
    $\Otilde(n^{2.5794} + n^2/\varepsilon)$.
\end{restatable}

Note that, for any constant $0 < \varepsilon < 1/3$, our
combinatorial $(1+\varepsilon)$-approximate combinatorial FDO 
from \autoref{thm:single_failure} is near-optimal with respect to all parameters.
The $\Theta(m)$ space is near-optimal by \autoref{thm:space_lower_bound_single},
the query time is $\Otilde(1)$, and the $\Otilde(mn)$ preprocessing time comes within sub-polynomial
factors of the conditional lower bound by Henzinger~et~al.~\cite{HenzingerL0W17}.
Furthermore, when fast matrix multiplication is permitted, 
our algebraic preprocessing algorithm is even faster on dense graphs.
However, our FDO is randomized.

\textbf{\textsf{Question 3 - derandomization.}}
Can the construction of \autoref{thm:single_failure} be derandomized 
in the same asymptotic running time?
We answer this question partially in that we derandomize the approximation part of our algorithm.
When combined with the DSO of Bernstein and Karger~\cite{BeKa09}
this gives a deterministic combinatorial FDO.
For the derandomization, we adapt the framework of Alon, Chechik, and Cohen~\cite{AlonChechikCohen19CombinatorialRP}.
We identify a set of $O(n^{3/2})$ critical paths one needs to hit,
and show how to compute them in $O(mn)$ time.
It is then enough to let the folklore greedy algorithm compute a hitting set in
 $\Otilde(n^2)$ time.

It remains an open problem whether one can derandomize the algebraic approach,
whose randomization stems solely from the DSO by Gu and Ren~\cite{GuRen21ConstructingDSO_ICALP}.

\begin{restatable}{theorem}{derandomization}
\label{thm:derandomization}
	For every unweighted directed graph and $\varepsilon > 0$,
	there exists a deterministic combinatorial $(1+\varepsilon)$-approximate FDO
	that takes $O(m)$ space and has $\Otilde(mn + n^2/\varepsilon)$ preprocessing time
	and $O(1)$ query time.
\end{restatable}

\textbf{\textsf{Question 4 - space vs.\ approximation trade-off for multiple failures.}}
Finally, we consider the case of multiple edge failures and examine similar questions.
What is a the minimum size for an exact, respectively, approximate, diameter oracle in the presence of up to $f$ edge failures?
We again prove an information-theoretic lower bound and show that for arbitrary finite stretch $\sigma$, any $\sigma$-approximation diameter oracle requires $\Omega(fn)$ bits of space,
at least if the oracle can be queried also with sets $F$ that contain non-edges.

\begin{restatable}{theorem}{spacelowerboundmultiple}
\label{thm:space_lower_bound_multiple}
	Suppose $f < n$.
	Any $f$-FDO with finite stretch that can be queried also for non-edges	must take $\Omega(fn)$ bits of space on graphs with $n$ vertices.
\end{restatable}

We develop an efficient $f$-FDO whose space requirement almost matches 
the lower bound.
Our result adapts and improves a construction by Bilò et al.~\cite{BGLP16}.
Note that we use the $\Otilde$-notation to suppress polylogarithmic factors in $n$.

\begin{restatable}{theorem}{multiplefailures}
\label{thm:multiple_failures}
	For every undirected graph with non-negative edge weights,
	there exists a deterministic combinatorial $(f\,{+}\,2)$-approximate $f$-FDO 
	that takes $\Otilde(fn)$ space and has 
	$\Otilde(fm)$ preprocessing time and $\Otilde(f^2)$ query time.
\end{restatable}

Real-world networks are often described as having a small diameter,
dubbed as the ``small world property''~\cite{kleinberg2000navigation}.
Many graph models used to analyze social and communication networks
have provable polylogarithmic guarantees on the diameter,
e.g.\ Chung-Lu graphs~\cite{chung2002average},
hyperbolic random graphs~\cite{friedrich2018diameter}, or the preferential attachement model~\cite{Hofstad16RandomGraphs}. 
We show that on graphs with low diameter one can swap approximation for query time even for multiple failures, while still retaining efficient preprocessing time and a low space requirement
To achieve this, we combine fault-tolerant trees that where introduced by Chechik et al.~\cite{ChCoFiKa17} with the random graphs of Weimann and Yuster~\cite{WY13}.

\begin{restatable}{theorem}{lowdiameter}
\label{thm:multiple_failures_low_diameter}
	Let $f$ be a positive integer and $\delta = \delta(n,m) > 0$ a real number.
	For every undirected unweighted graph with diameter at most $n^{\delta/f}/(f{+}1)$,
	there exists a randomized combinatorial $f$-FDO 
	that takes $O(n^{2+\delta})$ space, has $O(2^f)$ query time, and
	with high probability 
	$\Otilde(f mn^{1+\delta} + f \nwspace n^{2+(2-1/f)\delta})$ preprocessing time.
	Using fast matrix multiplication instead, one can construct the FDO w.h.p.\ in time
	$\Otilde(fn^{\omega+\delta} + f \nwspace n^{2+(2-1/f)\delta})$.
\end{restatable}

If the diameter is in fact polylogarithmic and the number of failures is bounded by 
$f = o(\log n/\log\log n)$, we obtain the following corollary.

\begin{corollary}
\label{cor:low_diameter}
	Let $f = o(\log n/\log\log n)$.
	For every undirected graph with polylogarithmic diameter, there is an
	$f$-FDO that takes $n^{2+o(1)}$ space and has $n^{o(1)}$ query time.
	It can be preprocessed in time $mn^{1+o(1)}$,
	or algebraically in time $n^{\omega+o(1)}$.
	If $f$ is constant, the preprocessing times
	are $\Otilde(mn)$, resp.\ $\Otilde(n^{\omega})$,
	with $\Otilde(n^2)$ space and $\Otilde(1)$ query time.
\end{corollary}

\subsection{Related Work}

We briefly review previous work on distance sensitivity oracles and diameter computation.

\textbf{\textsf{Distance sensitivity oracles.}}
Distance oracles for all-pairs distances were introduced 
in a seminal paper by Thorup and Zwick~\cite{ThorupZ01}. Demetrescu et al.~\cite{DeThChRa08}
extended the notion of distance oracles to the fault-tolerant setting in which either an edge or a vertex of a graph can fail (i.e., distance sensitivity oracles or DSOs).
They showed that it is possible to preprocess a directed weighted graph in $\Otilde(mn^2)$ time
to compute a data-structure of size $O(n^2 \log n)$ capable of answering distance queries in constant time. Bernstein and Karger~\cite{BeKa09} improved the preprocessing time to
$\Otilde(mn)$ and Duan and Zhang~\cite{DuanZhang17ImprovedDSOs} reduced the space to $O(n^2)$, which is asymptotically optimal.

Duan and Pettie~\cite{DP09} considered the more involved case of
two failures and presented an oracle with $O(n^2 \log^3 n)$ size, $O(\log n)$ query time
and polynomial construction time.
Chechik et al.~\cite{ChCoFiKa17} presented a DSO of size $O(n^{2+o(1)})$
that supports up to $o(\log n/\log \log n)$ edge failures and guarantees a stretch of $(1+\epsilon)$, for every constant $\epsilon > 0$. 
The approach has been recently extended to also handle vertex failures by 
Duan, Gu, and Ren~\cite{DuanGR21}.

The construction of DSOs have also been considered in the approximate regime~\cite{ChCoFiKa17}.
Algebraic algorithms are known to improve the preprocessing times,
if one is willing to employ fast matrix multiplication (for e.g., see~\cite{GuRen21ConstructingDSO_ICALP, ChCo20} and the references therein). 

\textbf{\textsf{Diameter computation.}}  The fastest known combinatorial algorithms (up to polylogarithmic factor) for both solving the all-pairs shortest paths (APSP) problem and the diameter problem, are the trivial ones with $\Otilde(mn)$ running time. There is extensive research on developing faster approximate APSP algorithms \cite{BaKa_SIGCOMP10, CoZw_JOA01, Kavitha12}, as well as faster approximation algorithms for the diameter \cite{Chechik:2014,Roditty16ApproxDiamEncycAlg}. For special classes of graphs, for example planar graphs, efficient exact algorithms for computing the diameter are known~\cite{GawrychowskiKMS18}.

\section{Preliminaries}
\label{sec:prelims}

We let $G=(V,E)$ denote the (possibly directed) base graph on $n$ vertices and $m$ edges.
We tacitly assume that $G$ is (strongly) connected, i.e., $m = \Omega(n)$.
For a graph $H$,
we denote by $V(H)$ the set of its vertices, and by $E(H)$ its edges.
The \emph{(closed) neighborhood} of a vertex $v \in V(H)$ is the set 
$N[v] = \{u \in V(H) \mid \{v,u\} \in E(H)\} \cup \{v\}$.
Let $P$ be a path in $H$, its \emph{length} $|P|$ is the number of its edges.
For any two vertices $x,y \in V(P)$, $P[x..y]$ is the subpath of $P$ from $x$ to $y$. 
For $s,t \in V(H)$, the \emph{distance} $d_H(s,t)$ 
is the minimum length of an $s$-$t$-paths in $H$;
if $s$ and $t$ are disconnected, we set $d_H(s,t) =+ \infty$.
We drop the subscript when talking about the base graph $G$.
The \emph{eccentricity} of $s$ is $\ecc(s,H) = \max_{t \in V(H)} d_H(s,t)$ and
the \emph{diameter} is $\diam(H) = \max_{s \in V(H)} \ecc(s,H)$.
Any graph distance can be stored in a single machine word on $O(\log n)$ bits.
Unless explicitly stated otherwise, we measure the space complexity in the number of words.
For a collection $F \subseteq \binom{V(H)}{2}$ of 2-sets of vertices (edges or non-edges),
let $H-F$ be the graph obtained from $H$ by removing all edges in $F$
(graph $H$ is not altered if $F \cap E(H) =\emptyset$).
A \emph{replacement path} $P_H(s,t,F)$ is a shortest path from $s$ to $t$ in $H-F$.
Its length $d_H(s,t,F) = |P_H(s,t,F)|$ is the \emph{replacement distance}.
The \emph{fault-tolerant diameter} of $H$ with respect to $F$ is the diameter of $H\,{-}\,F$.

For a positive integer $f$, an \emph{$f$-fault-tolerant diameter oracle} ($f$-FDO) for the graph $G$
is a data structure that reports, upon query $F$ with $|F| \le f$, the value $\diam(G\,{-}\,F)$.
For any $\sigma = \sigma(n,m,f) \ge 1$, such an oracle is $\sigma$\emph{-approximate},
or has \emph{stretch} $\sigma$,
if it answers a query $F$ with a value $\widehat D$ such that
$\diam(G\,{-}\,F) \le \widehat D \le \sigma \cdot \diam(G\,{-}\,F)$.
In case of a single failure, we write FDO for 1-FDO and abbreviate $F = \{e\}$ to $e$.
An \emph{$f$-distance sensitivity oracle} ($f$-DSO) reports,
upon query $(s,t,F)$ with $|F| \le f$, the replacement distance $d(s,t,F)$.

\subsection{(Mostly) Known FDOs for Single Edge Failures}
\label{subsec:prelims_trivial}

The first folklore FDO handles single edge failures in unweighted (directed or undirected) graphs.
It has also been observed in~\cite{HenzingerL0W17}.
The DSO of Bernstein and Karger~\cite{BeKa09}
constructible in $\Otilde(mn)$ time and
is able to report in constant time the exact distance of any pair of vertices 
in the presence of a single edge failure.
With this one can construct the FDO by explicitly computing all the eccentricities $\ecc(v,G-e)$,
for every vertex $v$ and every edge $e$ of $G$, in $O(n^3)$ time.
For a fixed  vertex $v$, the $m$ values $\ecc(v,G\,{-}\,e)$ can be obtained in $O(n^2)$ time as follows.
First compute a shortest paths tree $T_v$ of $G$ rooted at $v$.
For each edge $e$ that is not in $T_v$, we have that $\ecc(v,G\,{-}\,e)=\ecc(v,G)$.
For the tree-edges $e$ in $T_v$, we use the DSO to compute $\ecc(v,G\,{-}\,e)$
which is the maximum distance from $v$ to any other vertex in $G\,{-}\,e$. 
Therefore, $\ecc(v,G\,{-}\,e)$ can be computed by performing $n-1$ queries, 
as there are $n - 1$ edges in $T_v$, we need $O(n^2)$ time.
The fault-tolerant diameter $\diam(G-e)$ is the maximum of the $\ecc(v,G-e)$,
it can be stored in $O(m)$ space with one entry for each edge $e$.

The second folklore FDO can only be used for undirected edge-weighted graphs. The FDO has stretch $2$ and
uses the fact that the diameter of the graph is intimately related to the eccentricity of \emph{any} vertex.
For an arbitrary $v$, we have that $\ecc(v,G) \leq \diam(G) \leq 2\ecc(v,G)$ as, by the triangle inequality, we can bound the distance between any two vertices $u,u' \in V$ by $d_G(u,u') \leq d_G(u,v)+d_G(v,u') \leq 2\ecc(v,G)$. 
The FDO again computes a shortest paths tree $T$ rooted at a fixed source $v$
and stores an array of length $n-1$, corresponding to the edges of $T$.
For every such edge $e$, one computes and stores $2\ecc(s,G-e)$.
When queried with edge $e$, the FDO returns the stored value or,
if $e$ is not in the tree, the value $2\ecc(s,G)$.
The size of this FDO is $O(n)$.

A maybe lesser-known way of building FDOs is via spanners.
For any $\sigma > 0$, we say that a subgraph $H$ of $G$ is a \emph{spanner} of stretch $\sigma$ if, for every two vertices $s,t$ of $G$, we have $d_H(s,t) \leq \sigma d_G(s,t)$.
For every positive integer $k$, it is known how to construct a spanner $H$ of $G$ such that (a) $H$ has a stretch of $2k-1$ and (b) the size of $H$ is $O(n^{1+1/k})$~\cite{Althofer93SparseSpanners}. Observe that for every edge $e=\{u,v\}$
that is in $G$ but not in $H$, we have $d(u,v,e) \leq 2k-1$.
We now describe how spanners can be used to construct another easy oracle for undirected unweighted graphs whose stretch guarantee depends on both $k$ and the inverse of $\diam(G)$.
This implies that the oracle already performs quite well for large-diameter graphs.

We construct such a spanner oracle with parameter $k$ by first computing a spanner that satisfies (a) and (b). Then, we associate the value $\diam(G-e)$ to each edge $e$ in the spanner $H$ and build a dictionary in which we store information about the edges of the spanner together with the corresponding associated values. 
Consider a query of edge $e$. 
If $e \in E(H)$ the we return the value associated with $e$; otherwise, we return $\diam(G)+2(k-1)$.
The proof of the next lemma is deferred to \autoref{app:prelims_proof}.

\begin{restatable}{lemma}{spanneroracles}\label{lm:spanneroracle}
For every positive integer $k$, the spanner oracle with parameter $k$ has $O(n^{1+1/k})$ size, a constant query time, and a stretch of $1+2(k-1)/\diam(G)$. 
\end{restatable}

The result of \Cref{lm:spanneroracle} already implies the existence of sparse FDOs of $o(m)$ size and of stretch $\sigma < 3/2$ for sufficiently dense graphs with diameter strictly larger than 4. This does not contradict the lower bound of \Cref{thm:space_lower_bound_single}, but allows us to conclude that strong lower bounds on the size of FDOs for unweighted undirected graphs can only hold when the diameter of the input graph is bounded by a small constant.

\section{Single Edge  Failures}
\label{sec:single_failure}

First, we treat single edge failures, $f = 1$.
In this section, we assume the base graph $G$ to be directed
and present an $(1+\varepsilon)$-approximate fault-tolerant diameter oracle
with space $O(m)$ and $O(1)$ query time.
We give two variants, one is deterministic and combinatorial,
the other randomized and algebraic.
We then show that the space requirement is optimal up to the size of the machine word.

\subsection{An $(1+\varepsilon)$-approximate FDO for Single Failures}
\label{subsec:single_failure_algorithm}

We construct here the approximate FDO, thereby proving \autoref{thm:single_failure}.
Suppose we know for each $s,t \in V$ some shortest path $P(s,t)$ in $G$ 
and additionally have access to a distance sensitivity oracle
that, for any edge $e$, reports in constant time the replacement distances $d(s,t,e)$ whenever needed.
Clearly, $d(s,t,e)$ differs from the original graph distance only if $e$ is on $P(s,t)$.
To determine the diameters of \emph{all} the graphs $G\,{-}\,e$,
it is thus enough to query the DSO only for the edges on the shortest paths,
which can be done in time $O(n^2 \cdot \diam(G))$.
We use approximation to avoid the cubic running time in case of a large diameter.
For this, we randomly sample a small set $B$ of so-called \emph{pivots} and prove that it is enough to compute
the replacement distances only between pairs from $B \times V$,
instead of all pairs of vertices.
Subsequently, we derandomize the pivot selection.

We fill in the details
starting with the APSP computation in $G$ and the preprocessing of the DSO.
The combinatorial version uses a breath-first search from every vertex
and the DSO of Bernstein and Karger~\cite{BeKa09}, taking total time $\Otilde(mn)$.
Alternatively, compute APSP algebraically
and use the randomized DSO by Gu and Ren~\cite{GuRen21ConstructingDSO_ICALP}.\footnote{
	The DSO by Gu and Ren~\cite{GuRen21ConstructingDSO_ICALP} is not path-reporting;
	if it were, we would not have to compute APSP.
	The fastest path-reporting algebraic DSO was given by Ren~\cite{Ren20,Ren20_arxiv}
	and can be constructed in time $O(n^{2.7233})$ on directed graphs,
	respectively in time $O(n^{2.6865})$ on undirected graphs.
}
APSP is computable in time $O(n^{2.575})$ on unweighted directed graphs
with a variant of Zwick's algorithm~\cite[Corollary~4.5]{Zwick02DirectedAPSP},
this is in turn dominated by the $O(n^{2.5794})$ preprocessing time 
of the DSO~\cite{GuRen21ConstructingDSO_ICALP}.
After these computation, the distances $d(s,t)$, shortest paths $P(s,t)$ in $G$, 
and the replacement distances $d(s,t,e)$ are available to us (w.h.p., in the randomized case)
with a constant query time per distance/path edge.

From here on out, the process for both variants is the same.
Our fault-tolerant diameter oracle also allows non-edges to be queried, for which we return 
the original diameter $\diam(G)$.
To account for this, we store all edges in a static dictionary of size $O(m)$ 
that allows for worst-case constant look-up times after an $\Otilde(m)$
preprocessing~\cite{AlNa96,HagerupMiltersenPagh01DeterministicDictionaries}.\footnote{%
	The weak non-uniformity mentioned in~\cite{HagerupMiltersenPagh01DeterministicDictionaries},
	i.e., the need of compile-time constants depending on the word size,
	only holds if this size is $\omega(\log n)$,
	which is not the case for us.}

Now fix a parameter $\varepsilon > 0$ for the approximation, possibly even depending on $m,n$.
We initialize an array $D$ indexed by the edges of $G$,
all its cells hold the value $\diam(G)$.
Assume first that $\varepsilon \cdot \diam(G) = O(\log n)$.
For any two vertices $s,t$ and edge $e$ on the shortest path $P(s,t)$,
we update $D[e]$ to the maximum of the previous value and $d(s,t,e)$.
This takes $O(n^2 \diam(G)) = \Otilde(n^2/\varepsilon)$ time. 
After all updates, the entry $D[e]$ 
stores the \emph{exact} fault-tolerant diameter $\diam(G\,{-}\,e)$ (possibly w.h.p.).
For $\varepsilon \cdot \diam(G) = \omega(\log n)$,
we first give a randomized $(1{+}\varepsilon)$-approximation
and later derandomize it in \Cref{subsec:single_failure_derandomization}.
This yields the deterministic combinatorial algorithm of~\autoref{thm:derandomization}.
The remaining use of randomness in the algebraic variant is due to the DSO by Gu and Ren~\cite{GuRen21ConstructingDSO_ICALP}.

To guard for the case that the failure of $e$ disconnects the graph,
we compute all \emph{strong bridges} of $G$, 
that is, edges whose removal increases the number of strongly connected components,
in time $O(m)$ with the algorithm by Italiano, Laura, and Santaroni~\cite{Italiano12FindingStrongBridges}.
For each strong bridge $e$, we set $D[e] = \infty$.
To compute the other entries,
we construct the set $B \subseteq V$ of pivots by randomly sampling every vertex
independently with probability $C (\log n)/ (\varepsilon \diam(G))$
for a sufficiently large constant $C > 0$.
A simple calculation using Chernoff bounds shows that
$|B| = \Otilde(n/(\varepsilon \diam(G)))$ w.h.p.
Moreover, with high probability for all $s, t \in V$ and $e \in E$ 
such that $\varepsilon \diam(G) < d(s,t,e) < \infty$,
there exists a replacement path from $s$ to $t$ that avoids $e$
and additionally contains a pivot from $B$.
See~\cite{GrandoniVWilliamsFasterRPandDSO_journal,RodittyZwick12kSimpleShortestPaths} for details.
We update the entries of $D$ in the same fashion as above,
but now only use the (directed) distance $d(x,t,e)$ for all pivots $x \in B$ and vertices $t \in V$.
In the end, we add $\varepsilon \diam(G)$ to the value in $D[e]$.
The array $D$ is computable in time $O(m + n \nwspace |B| \diam(G)) = \Otilde(n^2/\varepsilon)$.

We verify that $D[e]$ is an $(1{+}\varepsilon)$-approximation 
of the fault-tolerant diameter $\diam(G\,{-}\,e)$.

\begin{restatable}{lemma}{singlefailureupperbound}
\label{lem:single_failure_upper_bound}
	We have $\diam(G\,{-}\,e) \le D[e] \le (1{+}\varepsilon) \diam(G\,{-}\,e)$ w.h.p.
\end{restatable}

\begin{proof}
	We can assume that $G\,{-}\,e$ is strongly connected 
	as otherwise $D[e] = \infty = \diam(G\,{-}\,e)$.
	The upper bound follows from
	$D[e] = \max_{x \in B, t \in V} d(x,t,e) + \varepsilon \diam(G) 
		\le (1{+}\varepsilon) \diam(G\,{-}\,e)$.
	
	The main part consists of showing the lower bound $D[e] \ge \diam(G\,{-}\,e)$.
	The idea is to prove the existence of a pivot $x \in B$
	and vertex $t \in V$ whose replacement distance underestimates the fault-tolerant diameter by
	at most an additive term $\varepsilon \diam(G)$,
	which we offset when computing $D[e]$.
	If $\diam(G\,{-}\,e)  \le \varepsilon \diam(G)$ (which can only happen for $\varepsilon \ge 1$),
	the lower bound holds vacuously as we have $D[e] \ge \varepsilon \diam(G)$.

	Let thus vertices $s,t \in V$
	be such that $d(s,t,e) = \diam(G\,{-}\,e) > \varepsilon \diam(G)$.
	Since $G\,{-}\,e$ is strongly connected
	the diameter is finite and realized by some replacement path $P(s,t,e)$.
	In particular, we have $|P(s,t,e)| > \varepsilon \diam(G)$.
	Let $y$ be the unique vertex on $P(s,t,e)$
	with $d(s,y,e) = \varepsilon \diam(G)$.
	Recall that w.h.p.\ the set $B$ hits \emph{some} shortest path $P'$ 
	from $s$ to $y$ that avoids $e$.
	The path $P'$ is not necessarily equal to the subpath $P(s,t,e)[s..y]$,
	but they have the same length $d(s,y,e)$.
	Substituting $P'$ for $P(s,t,e)[s..y]$ therefore guarantees
	a replacement path from $s$ to $t$ that (w.h.p.) has a pivot $x \in B$
	on its prefix of length $\varepsilon \diam(G)$.
	For notational convenience, we use $P(s,t,e)$ to also denote this particular path.
	
	The replacement distance from pivot $x$ to target $t$ satisfies
	$d(x,t,e) = |P(s,t,e)[x..t]| = |P(s,t,e)| - |P(s,t,e)[s..x]|
		\ge d(s,t,e) - \varepsilon \diam(G)$.
	The entry $D[e]$ is also updated using the pivot $x$,
	whence $D[e] \ge d(x,t,e) +  \varepsilon \diam(G) \ge d(s,t,e) = \diam(G\,{-}\,e)$.
\end{proof}

\subsection{Derandomization}
\label{subsec:single_failure_derandomization}

For the randomized combinatorial FDO, we had a
preprocessing time of $\Otilde(mn + n^2/\varepsilon)$.
The underlying APSP computation and the DSO are deterministic.
We now derandomize the approximation part
in the same asymptotic running time,
proving \autoref{thm:derandomization}.
In \autoref{lem:single_failure_upper_bound}, we used that the set $B$
intersects at least one long replacement path from $s$ to $t$ exactly.
We argue that it is in fact enough to hit
the set of all vertices with distance at most $\varepsilon \diam(G)$ from $s$
in each strongly connected $G\,{-}\,e$.
The pivot $x$ does not need to be on any replacement path.
The only assertion of \autoref{lem:single_failure_upper_bound}
that is possibly in doubt is the lower bound $D[e] \ge \diam(G\,{-}\,e)$.
Let again $s$ and $t$ be such that $d(s,t,e) = \diam(G\,{-}\,e)$ and let $x \in B$
be a pivot with $d(s,x,e) = d_{G\,{-}\,e}(s,x) \le \varepsilon \diam(G)$.
Whenever $G\,{-}\,e$ is strongly connected,
a replacement path $P(x,t,e)$ exists and, by the triangle inequality, we have 
$d(x,t,e) \ge d(s,t,e) - d(s,x,e) \ge d(s,t,e) - \varepsilon \diam(G)$.
The claim follows.

For the derandomization, we adopt the framework of Alon, Chechik and Cohen~\cite{AlonChechikCohen19CombinatorialRP}.
This involves efficiently finding a small set of critical paths
such that hitting them ensures to hit each $(\varepsilon \diam(G))$-ball 
in the strongly connected $G\,{-}\,e$.
If the critical paths are both short enough and few in numbers,
it is then enough to compute the hitting set via the folklore greedy algorithm.
In \cite{AlonChechikCohen19CombinatorialRP}, it was sufficient to give a single set of critical paths.
We generalize this to multiple sets, where the later-defined sets depend on the paths in the former.

Set $\ell  = \min \{\varepsilon \diam(G), \sqrt{n}\}$
and let $r$ be an arbitrary vertex in $G$.
We compute the in-tree $T_{\text{in}}(r)$, containing the shortest paths in $G$ leading to $r$,
with breath-first search.
In the set $\mathcal{P}$, we collect, for each vertex $s$ with $d(s,r) > \ell$,
the path of $T_{\text{in}}(r)$ starting in $s$ and having length $\ell$.
Let $P \in \mathcal{P}$ be a path with start vertex $s$ 
and let $e \in E(P)$ be such that it is not a strong bridge.
We compute the in-tree $T_{\text{in},e}(r)$ in $G\,{-}\,e$ rooted in $r$.
Note that $s$ has distance $d(s,r,e) \ge d(s,r) > \ell$ from the root in the tree.
We add the corresponding path to the set $\mathcal{P}_e$.
The original in-tree $T_{\text{in}}(r)$ contains only $n-1$ edges,
so all trees can be computed in total time\footnote{%
    For a single source, there are \emph{randomized} algorithms known
    that compute the trees faster~\cite{ChechikCohen19SSRP_SODA,ChMa19, GW12}.
}
$O(mn)$.
Moreover, there are at most $\ell+1$ paths with starting vertex $s$.
In total, we thus have $O(n\ell)$ paths each of length $\ell$.
A greedy algorithm computes a hitting set $B$ for all paths in the $\mathcal{P}$ and $\mathcal{P}_e$.
It iteratively selects the vertex that is contained in the most yet unhit paths,
it terminates in time $\Otilde(n\ell^2) = \Otilde(n^2)$ 
and produces a set of $|B| = \Otilde(n/\ell) = \Otilde(n/(\varepsilon \diam(G))$ pivots,
see~\cite{AlonChechikCohen19CombinatorialRP,King99FullyDynamicAPSP}.
We used the definition $\ell  = \min \{\varepsilon \diam(G), \sqrt{n}\}$
for both estimates.
Finally, we add the root $r$ to the set $B$ 
to cover all paths in the trees that are shorter than $\ell$.

\begin{lemma}
\label{lem:deterministic_HS}
	For each vertex $s \in V$ and edge $e$ such that $G\,{-}\,e$ is strongly connected,
	there exists a pivot $x \in B$ with $d(s,x,e) \le \varepsilon \diam(G)$.
\end{lemma}

\begin{proof}
	If $d(s,r) \le \varepsilon \diam(G)$, we are done.
	Otherwise, let $P$ be the prefix of length $\ell$
	of the path from $s$ to $r$ in the tree $T_{\text{in}}(r)$,
	whence $P \in \mathcal{P}$.
	If $P$ does not contain the edge $e$, it also exists in $G\,{-}\,e$ and the corresponding 
	pivot $x \in B \cap V(P)$ satisfies $d(s,x,e) = d(s,x) \le \ell \le \varepsilon \diam(G)$.
	If $P$ contains $e$, then let instead $P' \in \mathcal{P}_e$
	be the length-$\ell$ prefix of the path from $s$ to $r$ in $T_{\text{in},e}(r)$.
	Again, $x \in B \cap V(P')$ implies $d(s,x,e) \le \varepsilon \diam(G)$.
\end{proof}

\subsection{Space Lower Bounds}
\label{subsec:single_failure_lower_bounds}  

Finally, we prove \Cref{thm:space_lower_bound_single} thus showing that the space requirement
of the FDOs in \Cref{thm:single_failure,thm:derandomization} is near-optimal
provided that the stretch is $\sigma = \sigma(m,n) < 3/2$, 
that is, $\varepsilon < 1/2$.
This even holds for the simpler task of computing the diameter in undirected graphs.
For better exposition, 
we first show that any diameter oracle with such a stretch
requires $\Omega(n^2)$ space on at least one $n$-vertex graph,
which is, however, only tight for dense graphs.
We then sparsify the construction to for an $\Omega(m)$ bound for graphs with $m$ edges.
Any \mbox{$\sigma$-approximate} FDO solves the promise problem of distinguishing,
for each edge $e$, 
whether $G\,{-}\,e$ has diameter $2$ or $3$.

\begin{restatable}{lemma}{lowerbounddense}
\label{lem:space_lower_bound_dense}
	There is a graph $G$ on $n$ vertices 
	such that $G\,{-}\,e$ has diameter $2$ or $3$ for any $e \in E$.
	Any data structure that decides
	which one is the case must take $\Omega(n^2)$ bits of space.
\end{restatable}

\begin{proof}
	We give an incompressibility argument by encoding
	any binary $(n/4) \nwspace{\times}\nwspace (n/4)$ matrix $X$ in the fault-tolerant diameters of $G$.
	No data structure can store this in $o(n^2)$ bits.
	The construction is illustrated in \autoref{fig:lower_bound_dense}.
	
	Without loosing generality, $n$ is divisible by $4$,
	we can add up to three dummy vertices if needed.
	Split the vertex set equally into four groups $A$, $B$, $C$, $D$ and
	let $a_1, \dots, a_{n/4}$ be an arbitrary numbering of the elements of $A$,
	same with the other groups.
	All groups are made into cliques and, for all $i \in [n/4]$,
	we make $a_i$, $b_i$, and $c_i$ into a triangle.
	This results in matchings for the pairs $(A,B)$, $(B,C)$, and $(A,C)$, respectively.
	We further add edges so as to make $(B,D)$ into a biclique.
	To encode the matrix $X$, we introduce the edge $\{c_i, d_j\}$ if and only if $X_{i,j} = 1$.
	
	The graph $G$ indeed has diameter $2$ (even if $X$ is the all-zeros matrix).
	Vertices $a_i$ and $b_j$ are joined by the path $(a_i,a_j,b_j)$\emDash{}which
	by symmetry also holds for the other pairs of groups among $A$, $B$, or $C$\emDash{}and
	and the vertices $a_i$ or $c_i$ are connected to $d_j$ 
	via the paths $(a_i,b_i,d_j)$ or $(c_i,b_i,d_j)$, respectively.
	Removing any edge increases the diameter by at most $1$
	since for any $e = \{u,v\}$ there exists a common neighbor in $w \in N[u] \cap N[v]$.
	This is clear inside the (bi-)cliques.
	For the matching edges, say $e = \{a_i,b_i\}$, we have $w = a_j$, $j\neq i$.
	Finally, for $e = \{c_i,d_j\}$ (if it exists), we have $w = b_i$.

	We now prove that the graph $G-\{b_i,d_j\}$ has diameter $3$ if and only if
	the edge $\{c_i,d_j\}$ is \emph{not} present in $G$, that is, iff $X_{i,j} = 0$. 
	When arguing the diameter above, edge $\{b_i,d_j\}$ was only needed for the paths   
	$(a_i,b_i,d_j)$ and $(c_i,b_i,d_j)$.
	Consider the neighborhoods of the three vertices in $G-\{b_i,d_j\}$,
	$N[a_i] = A \cup \{c_i, d_i\}$,
	$N[c_i] = C \cup \{a_i,b_i\} \cup \{ d_k \mid X_{i,k} = 1\}$, and
	$N[d_j] = D \cup (B{\setminus}\{b_i\}) \cup \{ c_k \mid X_{k,j} = 1\}$.
	If $X_{i,j} = 1$, then the neighborhoods intersect, namely in $c_i$,
	keeping the diameter at $2$.
	If, however, $X_{i,j} = 0$, then $N[a_i] \cap N[d_j] = \emptyset$
	and the diameter increases to $3$. 
\end{proof}

\begin{figure}
  \captionsetup[subfigure]{justification=centering}
  \begin{subfigure}[t]{0.49\textwidth}
    \centering
    \includegraphics[scale=0.9,page=1]{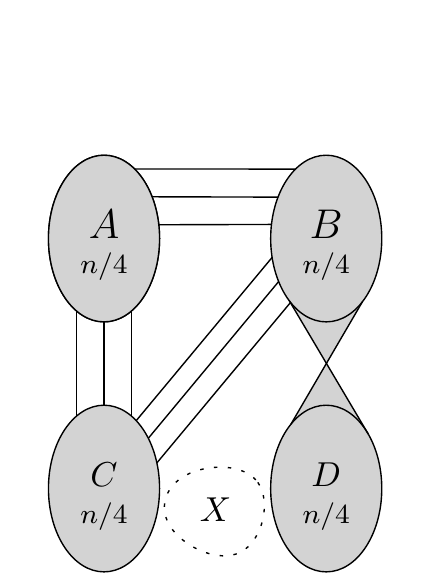}
    \caption{}
    \label{fig:lower_bound_dense}
  \end{subfigure} \hspace*{-.5cm}
  \begin{subfigure}[t]{0.49\textwidth}
    \centering
    \includegraphics[scale=0.9,page=2]{lower_bound_1}
    \caption{}
    \label{fig:lower_bound_sparse}
  \end{subfigure}
  \caption{Illustration of
  		\autoref{lem:space_lower_bound_dense} \textbf{\textsf{(\subref{fig:lower_bound_dense})}}
  		and of 
  		\autoref{lem:space_lower_bound_sparse} \textbf{\textsf{(\subref{fig:lower_bound_sparse})}}.
		The full ellipses $A$, $B$, $C$, $D$ are cliques on the respective number of vertices, 
		the dashed ellipse $R$ is an independent set.
		The three parallel lines stand for matchings, the two crossed lines for a biclique.
		Edges encoding the binary matrix $X$ run between $C$ and $D$.
		Every vertex $r \in R$ is connected to $a_1 \in A$, $b_1 \in B$, and $c_1 \in C$. 
  	}
  \label{fig:lower_bounds}
\end{figure}

We now refine the result to give a better bound for sparse graphs.
Note that a logarithmic gap remains between \autoref{lem:space_lower_bound_sparse}
and \autoref{thm:single_failure}
since we lower bound the space at $\Omega(m)$ bits while the FDO takes this many words.

\begin{restatable}{lemma}{lowerboundsparse}
\label{lem:space_lower_bound_sparse}
	There is a graph $G$ with $m$ edges 
	such that $G\,{-}\,e$ has diameter $2$ or $3$ for any edge $e \in E$.
	A data structure that decides
	which one is the case must take $\Omega(m)$ bits of space.
\end{restatable}

\begin{proof}
	The main weakness of the construction in \autoref{lem:space_lower_bound_dense}
	is that it requires $\Omega(n^2)$ edges inside the cliques.
	As it turns out, this is not necessary 
	and we can sparsify the graph $G$ as long as we keep its diameter at $2$.
	\autoref{fig:lower_bound_sparse} shows the idea of the sparsification.
	
	Let $m'$ be a parameter to be fixed later.
	We now store a binary $\sqrt{m'} \times \sqrt{m'}$ matrix $X$.
	Split the vertices into five groups,
	where $A$, $B$, $C$, $D$ each contain $\sqrt{m'}$ vertices and $R$ the remaining $n - 4\sqrt{m'}$.
	The edges among vertices in $A$ through $D$ 
	are the same as in \autoref{lem:space_lower_bound_dense}.
	Each vertex in $R$ has degree $3$ and is connected to $a_1$, $b_1$, and $c_1$.
	The graph $G$ has
	$4 \binom{\sqrt{m'}}{2} + 3\sqrt{m'} + m' +
		|\{ (i,j) \mid X_{i,j} = 1\}| + 3(n - 4\sqrt{m'}) = O(m')$ edges.
	We fix the parameter $m'$ such that the total number of edges is $m$.
	If needed, we introduce additional edges among vertices in $R$ without affecting the result.
		
	Note that the eccentricity of any vertex in $r \in R$ is $2$ 
	(even if $R$ is not an independent set).
	Vertex $a_i$ is reached via the path $(r,a_1,a_i)$,
	similar for the vertices in $B$ and $C$, the ones in $D$ are reached via $b_1$.
	Moreover, for any edge involving $r$, say $\{r,a_1\}$, we have $b_1 \in N[r] \cap N[a_1]$.
	Therefore, the proof that $G$ has diameter $2$, $G-e$ has diameter $2$ or $3$,
	and $G-\{b_i,d_j\}$ has diameter $3$ iff $X_{i,j} = 0$ is almost exactly
	as in \autoref{lem:space_lower_bound_dense}.
	The sole difference is the case in which the edge $\{b_1,d_j\}$ fails
	since this may also increase the eccentricity of $r$.
	This is settled by observing that the neighborhood $N[r] = \{r,a_1,b_1,c_1\}$ in $G-\{b_1,d_j\}$
	intersects $N[d_j]$ iff $X_{1,j} = 1$.
	To accommodate all possible matrices $X$, we require $\Omega(m') = \Omega(m)$ bits.
\end{proof}

The same construction shows that for edge-weighted graphs
there is no $(2{-}\varepsilon)$-approximate FDO, 
for any $\varepsilon = \varepsilon(m) > 0$,
with space $o(m)$.
In more detail, we choose an $\varepsilon' > 0$ small enough so that 
$\varepsilon' < 2\varepsilon/(1-\varepsilon)$
and give weight $\varepsilon'$ to all matching edges as well as the edges incident to vertices in $R$,
all other edges are weighted $2$.
One can verify that $\diam(G) = 2 + \varepsilon'$
and the fault-tolerant diameter $\diam(G - \{b_i,d_j\})$ remains at that value
iff $\{c_i,d_j\}$ is present, it raises to $4+\varepsilon'$ otherwise.
The bound on the stretch cannot be improved as shown by the trivial FDO discussed in the introduction, which gives a $2$-approximation in $O(n)$ space.

\section{Multiple Edge Failures}
\label{sec:multiple_failures}

We now turn to multiple edge failures.
Recall that in the fault-tolerant setting the maximum number $f$ of failures is known in advance,
and stretch, space, preprocessing, and query time usually depend on $f$.
In this section, we first prove the following lemma.
Let $\alpha = \alpha(m,n)$ denote the inverse Ackermann function.

\begin{lemma}[\autoref{thm:multiple_failures} with explicit logarithmic factors]
	For every undirected graph with non-negative edge weights,
	there exists a deterministic combinatorial $(f\,{+}\,2)$-approximate $f$-FDO 
	that takes $O(fn\log^2 \! n)$ space and has 
	$O(fm \nwspace \alpha+fn\log^3 \! n)$ preprocessing time and $O(f^2 \log^2 \! n)$ query time.
	For $f=1$, the size of the oracle is $O(n)$,
the preprocessing time $O(m \nwspace \alpha+n\log n)$, and the query time is constant.
\end{lemma}

Bilò et al.~\cite{BGLP16} designed an $(2f{+}1)$-approximate single-source $f$-DSO.
That means, the oracle processes an undirected graph $G$ with non-negative edge weights and a distinguished source $s$, and, upon query $(t,F)$ with $|F| \le f$, it returns $d(s,t,F)$.
The oracle can be built in $O\big(fm\alpha + fn \log^3 n\big)$ time, has size $O(fn \log^2 n)$, 
and answers queries in $O(f^2 \log^2n)$ time.
In principle we can modify the oracle so as, when queried with the set $F$, 
it returns twice the eccentricity of $s$ in the graph $G\,{-}\,F$. This would clearly allow us to construct an $f$-FDO of stretch  $2{\cdot}(2f{+}1)$.
We show that the same oracle construction, but with a better query algorithm, allows us to develop an $f$-FDO of stretch $f+2$.

We let $w(e)$ denote the weight of the edge $e \in E$.
The length of a path is now defined as the sum of its edge weights; 
the definitions of distance and diameter are adjusted accordingly.
The oracle in~\cite{BGLP16} first computes a shortest path tree $T$ of $G$ 
rooted at the source $s$ and uses it to re-weight all the edges of $G$. 
The new weight function $w'$ assigns weight of $0$ to each edge of $T$
and weight $w'(e)=d(s,x)+w(e) + d(y,s)$ to any other edge $e= \{x,y\}$.
When queried with $(t, F)$, the oracle computes a spanning forest $T_F$ of $G\,{-}\,F$ w.r.t.\ the new weight function $w'$ in $O(f^2 \log^2 n)$ time. Let $k=|F \cap E(T)|$. The oracle replaces the $k$ failing edges in $F \cap E(T)$ with a minimum-weight set of edges in $G\,{-}\,F$ w.r.t.\ to $w'$, say $E_F$,
whose addition to $T\,{-}\,F$ forms a spanning forest of $G\,{-}\,F$.\footnote{%
	This is done by computing, for each unordered pair $\phi=(T',T'')$ of connected components of $T-F$, 	the minimum-weight edge w.r.t. $w'$, say $e_{\phi}$, 
	that has one endpoint in $T'$ and the other endpoint in $T''$. 
	Then, the set $E_F$ is computed in $O(f^2)$ time using any time-efficient algorithm 
	for computing a minimum spanning tree of an auxiliary graph 
	in which each of the connected components of $T-F$ is modelled by a vertex
	and the edge between the unordered pair $\phi=(T',T'')$ of $T-F$ has a weight equal to $w'(e_\phi)$. 	The authors of~\cite{BGLP16} design a data structure that is able to retrieve, 
	for each pair $\phi=(T',T'')$ of connected components of $T-F$, 
	the edge $e_{\phi}$ in $O(\log^2 n)$ time.
}
The obtained forest $T_F$ is then used to estimate the distance from $s$ to $t$ in $G\,{-}\,F$.
We reuse a nice property proven in~\cite{BGLP16}.

\begin{lemma}[Bilò et al.~\cite{BGLP16}]
\label{lem:msf_of_H_minus_F_is_also_a_msf_of_G_minus_F}
$T_F$ is a minimum spanning forest of $G\,{-}\,F$ w.r.t.\ $w'$.
\end{lemma}

Our query algorithm works as follows. 
Let tree $T$ be rooted at $s$ and $F \cap E(T) = \{ f_1,\dots,f_k\}$ with $k\leq f$
the edges in $T$ that are also in $F$.
Let $T_0,\dots,T_k$ denote the $k+1$ subtrees of $T\,{-}\,F$,
and $r_i$ the root of the subtree $T_i$.
W.l.o.g., we assume $r_0=s$. We use $f_1,\dots,f_k$ to compute the roots $r_1,\dots,r_k$ in $O(k)$ time.
We then build a forest $T'$ on $k+1$ new vertices $v_0,\dots,v_k$, where $v_i$ represents $T_i$.
The forest $T'$ contains an edge $\{v_i,v_j\}$ 
iff $E_F$ contains an edge $e$ with one end point in $V(T_i)$ and the other in $V(T_j)$. 
Obviously, if $T'$ is not connected, then we can simply certify that $\diam(G\,{-}\,F)=\infty$.
So, we assume that $T'$ is a tree. 
We root $T'$ at $v_0$ and denote by $e_i$ the edge that joins $v_i$ with its parent $p(v_i)$.
We compute the value $\Delta=\max_{1 \le i \le k} w'(e_i)-d(s,r_i)$
and output $\widehat D=f\Delta + 2 \cdot \max_{t \in V}d(s,t)$.
The time needed for the query algorithm is dominated by the computation of $T$ in time $O(f^2 \log^2 n)$
as all the new operations can be performed in $O(f^2)$ time.
Observe that $\max_{t \in V}d(s,t)$ is independent of $F$ 
and can be precomputed in time $O(n)$.

For a single failure, $f=1$, the query time can be reduced to $O(1)$.
In fact, for each edge $e$ of $T$, it is enough to precompute the minimum weight edge of 
$E(G){\setminus}E(T)$, w.r.t. weight function $w'$, that crosses the cut induced by $T-e$.
This, a.k.a.\ the sensitivity analysis problem of a minimum spanning tree,
can be solved in $O(m\log \alpha)$ time on a graph with $m$ edges~\cite{Pettie15}.
We show in the remainder that $\widehat D$ is an $(f+2)$-approximation of $\diam(G-F)$.
The proof of the following lemma can be found in \autoref{app:proofs_multiple_failures}.

\begin{restatable}{lemma}{stacsalgorithm}
\label{lem:lower_bound_stacs_algorithm}
We have that $\diam(G-F)\geq \Delta$.
\end{restatable}

We now prove the approximation with the help of \autoref{lem:lower_bound_stacs_algorithm}.

\begin{lemma}
The value $\widehat D$ satisfies $\diam(G-F) \leq \widehat D \leq (f+2)\diam(G-F)$.
\end{lemma}

\begin{proof}
    Again, we only need to prove anything if $T'$ is connected, which implies that $T_F$ is connected. 
    By \autoref{lem:lower_bound_stacs_algorithm}, we have that $\diam(G\,{-}\,F) \geq  \Delta$.
    Moreover, $\diam(G\,{-}\,F)\geq \diam(G)\geq \max_{t \in V}d(s,t)$.
    The value $\widehat D$ returned by the query algorithm satisfies $\widehat D \leq f\Delta + 2\max_{t \in V}d(s,t) \leq (f\,{+}\,2)\diam(G-F)$.
    It remains to show that $\widehat D \geq \diam(G\,{-}\,F)$.
    We prove the latter by verifying that, for any two vertices $x$ and $y$,
    $\widehat D \geq d(x,y,F)$ holds. 
    
    Let $r_x$ and $r_y$ be the roots of the subtrees of $T\,{-}\,F$ 
    that contain $x$ and $y$, respectively.
    It is possible that $r_x=r_y$.
    Let $r_{p(i)}$ denote the root of the tree of $T\,{-}\,F$
    that corresponds to the parent vertex $p(v_i)$ in $T'$.
    Consider the subgraph of $T_F$ consisting of the edges of the paths in $T_F$ 
    between the following pairs of vertices: 
    (a) $r_{p(i)}$ and $r_i$ for every $i$, (b) $x$ and $r_x$, (c) $y$ and $r_y$. 
    The subgraph contains a path from $x$ to $y$ since $T_F$ is connected. 
    Therefore, the replacement distance $d(x,y,F)$ is upper bounded by the total weight of the subgraph.
    The path in $T_F$ between $r_x$ and $x$ has length at most $\max_{t \in V}d(s,t)$
    as $r_x$ is an ancestor of $x$ in the shortest path tree $T$ rooted at $s$;
    same for $r_y$ and $y$.
    Finally, for any $i > 0$, let $e_i=\{x_i,y_i\}$ be the edge in $E_F$ that caused the addition of the edge $(v_i,p(v_i))$ in $T'$.
    W.l.o.g., we assume that $x_i$ (resp., $y_i$) is a vertex of the tree of $T\,{-}\,F$ represented by $v_i$ (resp., $p(v_i)$) in $T'$.
    The path from $r_i$ to $r_{p(i)}$ in $G-F$ has length at most
    $d(r_i,x_i)+ w(e_i) + d(y_i,r_{p(i)}) \leq d(r_i,x_i) +w(e_i)+ d(y_i,s)+ d(s,r_i)-d(s,r_i)= w'(e_i)-d(s,r_i) \leq \Delta$. 
    Therefore, $d(x,y,F) \leq k\Delta + 2\max_{t \in V}d(s,t) \le f\Delta + 2\max_{t \in V}d(s,t) = \widehat D$.
\end{proof}

\subsection{Exact $f$-FDO for Low Diameter}
\label{subsec:multiple_failures_low_diameter}

We show that one can swap approximation for query time in low-diameter graphs,
namely, with diameter at most $n^{\delta/f}/(f{+}1)$ for arbitrary $\delta = \delta(m,n) > 0$.
This is summarized in \autoref{thm:multiple_failures_low_diameter}.
The case $f\,{=}\,1$ is solved like in \Cref{subsec:single_failure_algorithm}
only that there is no need for approximation here as the diameter is small enough
to process all pairs of vertices in time $O(n^{2+\delta})$.
We thus assume $f \ge 2$.
We adapt a space-saving technique introduced by Chechik et al.~\cite{ChCoFiKa17}.
In a bird's-eye view, we construct a recursion tree $T(s,t)$ of size $O(n^\delta)$
for each pair of vertices $s$ and $t$.
It contains all \emph{relevant} replacement distances $d(s,t,F)$ for sets $F$ with up to $f$ failures.
We then show how we can simulate the search for $\diam(G\,{-}\,F)$
in the $O(n^2)$ trees in total time $O(2^f)$.

Afek et al.~\cite[Theorem~1]{Afek02RestorationbyPathConcatenation_journal} showed 
that if $G$ is undirected,
then any shortest path in $G\,{-}\,F$, with $|F|$,
is a concatenation of at most $|F|+1$ shortest paths in $G$.
The condition on the diameter and $|F| \le f$ ensure that every path below
has length at most $(f{+}1) \cdot \diam(G) \le n^{\delta/f}$.

Assume we have access to a path-reporting $f$-DSO.
That means, upon query $(s,t,F)$,
the oracle either certifies that $d(s,t,F) = \infty$, i.e., $s$ and $t$ are disconnected in $G\,{-}F$,
or reports the replacement distance and a shortest $s$-$t$-path in $G\,{-}F$.
The preprocessing time of the combinatorial version is assumed to be $\Otilde(fmn^{1+\delta})$
with a $\Otilde(f n^{(1-1/f)\delta} +|P| \nwspace) = \Otilde(f n^{(1-1/f)\delta})$
query time w.h.p.\ reporting path $P$.
Here, we used the assumption $f \ge 2$, whence $|P| \le n^{\delta/f} = \Otilde(f n^{(1-1/f)\delta})$.
Alternatively, we have algebraic preprocessing in time $\Otilde(fn^{\omega+\delta})$.
We show how to obtain the oracle in \autoref{app:DSO_subroutine},
using an idea of Weimann and Yuster~\cite{WY13}
with a more refined analysis of the query time.

Fix two vertices $s$ and $t$.
We construct the tree $T(s,t)$ recursively.
Each node in the tree is associated with a set $F \subseteq \binom{V}{2}$ 
containing $f' = |F'| \le f$ possible failures.
We have $F' = \emptyset$ in the root.
Upon creation, the node queries the assumed oracle 
with $(s,t,F')$ and holds the returned path $P(s,t,F')$, if any.
If $f' = f$ or $s$ and $t$ are disconnected in $G\,{-}\,F'$, the node is a leaf.
Otherwise, it has $d(s,t,F')$ many children, one for each edge of $e \in E(P(s,t,F'))$ of the path.
The respective child is associated with the set $F' \cup \{e\}$.

The tree indeed has at least one node for every distinct replacement distance $d(s,t,F)$ 
with $|F| \le f$.
To see this, let $F',F$ be two sets with $F' \subseteq F \subseteq \binom{V}{2}$.
Clearly, we have $d(s,t,F') \le d(s,t,F)$,
but $d(s,t,F') < d(s,t,F)$ can only hold if $F{\setminus}F'$ contains an edge
of the path $P(s,t,F')$ in the node associated with $F'$.
%
The fan-out of each node is at most $n^{\delta/f}$, the height of the tree is $f$.
For all $s,t \in V$, the trees thus have $O(n^{2+\delta})$ nodes in total
and can be constructed with that many queries to the $f$-DSO in time $\Otilde(f n^{2+(2-1/f)\delta})$.

Consider the following naive algorithm to handle a query to the $f$-FDO
for the fault-tolerant diameter $\diam(G\,{-}\,F)$.
Each tree $T(s,t)$ is searched individually starting in the root.
The processing of a node depends on the associated set $F'$.
If it is a leaf or the set $F{\setminus}F'$ is disjoint from the replacement path $P(s,t,F')$,
then we return the length $d(s,t,F')$ of the path;
otherwise, we recurse on all children associated with $F' \cup \{e\}$ 
for all edges $e \in (F{\setminus}F') \cap E(P(s,t,F'))$.
By the argument as above,
the maximum over all reported distances is indeed 
$\max_{s,t \in V; F' \subseteq F} d(s,t,F') = \diam(G\,{-}\,F)$.
This approach can be improved significantly by aggregating the values $\{d(s,t,F')\}_{s,t \in V}$ already at construction.

Observe that we never query the underlying $f$-DSO with a set $F'$ that contains non-edges.
We prepare a hash table $H$ whose entries are indexed by subsets of $E$ of size at most $f$.
For every query $(s,t,F')$ we compare the returned replacement distance with the value $H[F']$.
If no such entry exists, we initialize it with $d(s,t,F')$;
else, we update it to $\max\{ H[F'], d(s,t,F')\}$.
The final table has size $O(n^{2+\delta})$ and we discard the trees.
The table $H$ is constructible w.h.p.\ in time $O(n^{2+\delta})$, 
guaranteeing constant query time~\cite{Dietzfelbinger94DynamicPerfectHashing,PaghRodler04CuckooHashing}.
However, to simulate the naive algorithm for the query $F$ to the $f$-FDO,
we have to check $H[F']$ for all $O(2^f)$ subsets $F' \subseteq F$
as we do not know which ones
were used during construction. 

\subsection{Space Lower Bound}
\label{subsec:multiple_failures_lower_bounds}  

We conclude with the space lower bound of \autoref{thm:space_lower_bound_multiple}.
It rules out any finite stretch in $o(fn)$ space for an arbitrary number $f$ of failures.
We use the fact that an $f$-FDO with finite stretch is able to decide 
whether the edges in $F$ are a cut-set of the graph.

	Assume for now that $f$ is even.
	Let $k$ be the largest integer such that $fk+1 \le n$.
	We construct a graph $G$ as follows.
	It has vertices $c, v_1, \dots, v_{fk}$ as well as $ n\,{-}\,fk\,{-}\,1$ auxiliary vertices.
	Define $E_i = \{\{v_i, v_j\} \mid 1 \le |i - j| \le f/2\}$.
	The edge set of $G$ is  $\bigcup_{i=1}^{fk} E_i$
	together with all possible edges $\{c,u\}$, including to the auxiliaries.
	In other words, $G$ consists of a star centered at $c$ with $n-1$ leaves,
	and leaves $v_i, v_j$ are joined by an edge iff their indices have difference at most $f/2$.
	Let set $\mathcal{G}$ contain all spanning subgraphs of $G$
	that retain at least all star edges incident to $c$.
	Since $|E_i| = f$,
	there are $|\mathcal{G}| = 2^{(f{-}1)f k/2} = 2^{\Omega(fn)}$ such subgraphs.

	Let $H$ be any subgraph in $\mathcal{G}$.
	For $i \neq j$ with $|i-j| \le f/2$,
	define the set $F_{i,j} = (E_i{\setminus}\{\{v_i,v_j\}\}) \cup \{\{c,v_i\}\}$.
	Note that $F_{i,j}$ may contain non-edges.
	We have $|F_{i,j}| = f$ and 
	evidently $\{v_i,v_j\}$ is present in $H$
	iff $H - F_{i,j}$ is connected.
	Any two $f$-FDOs for graphs in $\mathcal{G}$ thus differ in at least one bit.
	For odd values $f \ge 3$, we emulate this using $f-1$ failures.
	
	For the remaining case $f=1$, we use a different construction.
	W.l.o.g., $n$ is even,
	connecting a single excess vertex to some other vertex in the graph is immaterial.
	The graph $G$ contains two parallel paths $P_1$ and $P_2$, each on $n/2$ vertices,
	respectively numbered from $1$ to $n/2$. 
	The graph also contains a matching $M$
	in which the $i$-th vertex of $P_1$ is matched with the $i$-th vertex of $P_2$.
	Let $\mathcal{G}$ be the set of all spanning subgraphs
	that have at least all the edges of $P_1$ and $M$.
	We have $|\mathcal{G}|=2^{(n/2)-1} = 2^{\Omega(n)}$.
	Let $H \in \mathcal{G}$ and define $e_i$, with $i < n/2$, be the edge of $P_1$ between the $i$-th and $(i{+}1)$-th vertices. 
	The corresponding edge of $P_2$ is present in $H$ if and only $H-e_i$ is connected. 

\bibliographystyle{plainurl} 
\bibliography{FDO_bib}

\newpage
\appendix

\section{Proofs Omitted in \autoref{sec:prelims}}
\label{app:prelims_proof}

\spanneroracles*

\begin{proof}
The bounds on the query time and the size are by construction. We prove the upper bound on the stretch.
For every edge $e \in E(H)$ the oracle reports the exact value $\diam(G-e)$. Therefore, we only need to bound the approximation guarantee when the failing edge $e \not \in E(H)$. Let $e=(u,v)$. We have that $d(u,v,e)\leq 2k-1$.
As a consequence, any path of length $\ell$ that uses edge $e$ in $G$ has a length that is at most $\ell+2(k-1)$ in $G-e$ as we can bypass the edge $e$ by the path between $u$ and $v$ in $G-e$ of length at most $2k-1$. Therefore, $\diam(G-e) \leq \diam(G)+2(k-1)$. This implies that the value $\widehat D=\diam(G)+2(k-1)$ returned by the query oracle satisfies $\diam(G-e) \leq \widehat D \leq \diam(G)+2(k-1)$.
Therefore, using the fact that $\diam(G-e) \geq \diam(G)$, the stretch of the oracle is $1+2(k-1)/\diam(G)$.
\end{proof}

\section{Proofs Omitted in \autoref{sec:multiple_failures}}
\label{app:proofs_multiple_failures}

\stacsalgorithm*

\begin{proof}
If $\diam(G\,{-}\,F)=\infty$ there is nothing to show.
We thus assume that $T'$ is connected, which implies that also $T_F$ is connected.
Let $i^*$ be an index such that $\Delta=w'(e_{i^*})-d(s,r_{i^*})$.
We prove the lemma by showing that the replacement distance from the source to the root $r_{i^*}$ observes $d(s,r_{i^*},F) \geq w'(e_{i^*})-d(s,r_{i^*})$.
A replacement path $P(s,r_{i^*},F)$ (which exists as $T_F$ is a connected subgraph of $G\,{-}\,F$) 
crosses the cut induced by the removal of $e_{i^*}$ from $T_F$ with at least one edge, say $e = \{x,y\}$.
Let $x$ be in the connected component of $T_F\,{-}\,e_{i^*}$ as $r_{i^*}$. 
This implies  the following property:
\inlineequation[eq:lb_replacement_path]{d(s,r_{i^*},F) \geq d(s,y)+w(e)+d(x,r_{i^*})}.

By \autoref{lem:msf_of_H_minus_F_is_also_a_msf_of_G_minus_F}, $w'(e) \geq w'(e_{i^*})$ as otherwise we could replace $e_{i^*}$ by $e$ in $T_F$ and obtain a spanning tree of $G-F$ whose cost is strictly smaller than that of $T_F$.
Moreover, by the triangle inequality, $d(s,x) \leq d(s,r_{i^*})+d(x,r_{i^*})$, 
whence $d(x,r_{i^*})-d(s,x) \geq -d(s,r_{i^*})$. Starting from \Cref{eq:lb_replacement_path} and using both  inequalities, we can derive at the desired bound.
\begin{align*}
    d(s,r_{i^*},F)  & \geq d(s,y)+w(e)+d(x,r_{i^*}) = d(s,y)+w(e)+d(s,x)-d(s,x)+d(x,r_{i^*})\\
                & = w'(e)-d(s,x)+d(x,r_{i^*}) \geq w'(e_{i^*})-d(s,x)+d(x,r_{i^*})\\
                & \geq w'(e_{i^*})-d(s,r_{i^*}) =\Delta. \qedhere
\end{align*}
\end{proof}

\section{A Path-Reporting $f$-DSO with $\Otilde(f \nwspace n^{(1-1/f)\delta})$ Query Time}
\label{app:DSO_subroutine}

Recall that we assume that any shortest path in $G-F$ for any $F$ with $|F| \le f$ has at most $n^{\delta/f}$ many edges.
Also, the $f$-DSO is never queried with a set that contains non-edges.
We use a technique introduced by Weimann and Yuster~\cite{WY13}.
Let $k = Cf n^{\delta} \log n$ for a sufficiently large constant $C > 0$.
We create a set $G_1, \dots, G_k$ of spanning subgraphs of $G$.
For every $i$, $G_i$ is obtained by excluding any edge in $E$ independently with probability $n^{-\delta/f}$.
Combining the results in~\cite{WY13} with all replacement paths having length at most $n^{\delta/f}$ gives the following.

\begin{lemma}[Weimann and Yuster~\cite{WY13}]
\label{lem:Weimann_Yuster}
	With high probability, for all pairs of vertices $s,t \in V$ and sets $F \subseteq E$
	of at most $f$ edges, there exists an index $i \in [k]$ and a replacement path $P(s,t,F)$ 
	such that $P(s,t,F)$ is a shortest path in $G_i$.
\end{lemma}

\noindent
Along the same lines, we can bounds the number of graphs that exclude at least $f' \le f$.

\begin{lemma}
\label{lem:few_WY_graphs}
	Let $f' \le f$ be a positive integer.
	With high probability for all sets $F \subseteq E$ with $f' = |F|$,
	there are $\Otilde(f n^{(1-f'/f)\delta)})$ subgraphs $G_i$ 
	such that no edge of $F$ is in $G_i$.
\end{lemma}

\begin{proof}
	Let $k_F$ denote the number of subgraphs that exclude at least all of $F$.
	We have $\operatorname{E}[k_F] = k \nwspace (n^{-\delta/f})^{f'} = Cf \cdot n^{(1-f'/f)\delta} \ln n$.
	Let $N$ abbreviate $n^{(1-f'/f)\delta}$.
	Using Chernoff bounds (see e.g.~\cite{Mitzenmacher17Probablility}),
	we get that the probability of $k_F$ being more than double its expectation is
	$\operatorname{P}[k_e \ge 2\operatorname{E}[k_e] \nwspace] 
		\le \exp(-\operatorname{E}[k_F]/3) = n^ {-cfN/3}$.
	A union bound over the $O(n^{2f'})$ possible sets $F$ with $|F| = f'$ gives the lemma.
\end{proof}

\noindent
For each edge $e$,
we record during construction the set $S_e$ of graphs that exclude $e$.
Computing APSP in all the $G_i$ takes time $\Otilde(fmn^{1+\delta})$ combinatorially
or $\Otilde(fn^{\omega+\delta})$ algebraically.
In the same time bounds, we prepare an $k \times O(n^2)$ table
indexed by the subgraphs and pairs of vertices,
where the entry $[i,s,t]$ contains the distance $d_{G_i}(s,t)$.
Finally, we prepare, for each $G_i$, the information about its shortest paths
in the usual way of $n$ predecessor trees each.

Upon query $(s,t,F)$, $|F|\le f$, we first have to identify those graphs that contain no edge of $F$.
One could cycle to all graphs and check each in $O(f)$ time.
However, we can slightly improve on that using the sets $S_e$.
We intersect the sets for all edges in $F$
to obtain the set $S_F = \bigcap_{e \in F} S_e$ of precisely the indices we are looking for.
The intersection can be implemented such that it runs in time proportional
to the size of the smaller set.
By \autoref{lem:few_WY_graphs}, the size of \emph{all} intermediate sets $S_{F'}$
with $F' \subseteq F, |F'| = f'$
is bounded by $2C f n^{(1-f'/f)\delta} \ln n$ w.h.p.
Computing $S_F$ thus takes time linear in
\begin{equation*}
	\sum_{i = 1}^{|F|} 2C f n^{(1-i/f)\delta} \ln n 
		\le 2Cf n^{(1-1/f)\delta} (\ln n) \cdot \sum_{j=0}^\infty ({n^{-\delta/f}})^j 
		= \Otilde(f n^{(1-1/f)\delta}),
\end{equation*}

\noindent
where the last estimate is due to $\delta > 0$, whence ${n^{-\delta/f}} < 1$.

We retrieve the minimum of all values stored in entries $[i,s,t]$ with $i \in S_F$,
again in time $\Otilde(f n^{(1-|F|/f)\delta})$ w.h.p.
This is the correct replacement distance $d(s,t,F)$ w.h.p.\ by \autoref{lem:Weimann_Yuster}.
We return that minimum and, in case it is finite, a shortest $s$-$t$-path $P$ in $G_i$
for some index $i$ that assumes the minimum.
In total, the query time is $\Otilde(f \nwspace n^{(1-1/f)\delta} + |P|)$.

\end{document}